\documentclass[a4paper,UKenglish,cleveref, autoref]{lipics-v2019}
\usepackage{xspace}

\newcommand{\com}[1]{}

\newcommand{\calO}{\mathcal{O}}
\newcommand{\fr}{Fr\'{e}chet distance}

\usepackage{wrapfig}
\usepackage{lipsum} %remove
\usepackage{placeins}
\usepackage{tikz}
\usepackage{caption}
\usepackage{subcaption}
\usepackage{interval}
\usepackage{complexity}
\usepackage[noend,ruled,vlined,linesnumbered]{algorithm2e}
\SetFuncSty{textsc}
\SetDataSty{textit}
\SetCommentSty{textrm}
\SetKwFunction{Extract}{ExtractFirstDigit}
\SetKwData{Forward}{forward}
\SetKw{True}{true}
\SetEndCharOfAlgoLine{}

\usepackage{microtype}%if unwanted, comment out or use option "draft"

\title{Computing the Fr\'echet distance of trees and graphs of bounded tree width%\footnote{Supported by our friends.}
} 

\titlerunning{Computing the Fr\'echet distance of trees and graphs}%optional, please use if title is longer than one line

\author{Maike Buchin}{Ruhr-Universit\"at Bochum, Germany }{maike.buchin@rub.de}{}{}

\author{Amer Krivo\v{s}ija}{Department of Computer
        Science, TU Dortmund, Germany }{amer.krivosija@tu-dortmund.de}{}{}

\author{Alexander Neuhaus}{Department of Computer
        Science, TU Dortmund, Germany }{alexander2.neuhaus@tu-dortmund.de}{}{}        %TODO mandatory, please use full name; only 1 author per \author macro; first two parameters are mandatory, other parameters can be empty. Please provide at least the name of the affiliation and the country. The full address is optional

\authorrunning{M. Buchin, A. Krivo\v{s}ija and A. Neuhaus}%TODO mandatory. First: Use abbreviated first/middle names. Second (only in severe cases): Use first author plus 'et al.'

\Copyright{Maike Buchin, Amer Krivo\v{s}ija and Alexander Neuhaus}%TODO mandatory, please use full first names. LIPIcs license is "CC-BY";  http://creativecommons.org/licenses/by/3.0/

%TODO mandatory; please add comma-separated list of keywords

\relatedversion{}%optional, e.g. full version hosted on arXiv, HAL, or other respository/website
\funding{A. Krivo\v{s}ija was supported by the German Science Foundation (DFG) Collaborative Research Center SFB 876 "Providing Information by Resource-Constrained Analysis", project A2.}

\EventEditors{Gill Barequet and Yusu Wang}
\EventNoEds{2}
\EventLongTitle{35th International Symposium on Computational Geometry (SoCG 2019)}
\EventShortTitle{SoCG 2019}
\EventAcronym{SoCG}
\EventYear{2019}
\EventDate{June 18--21, 2019}
\EventLocation{Portland, United~States}
\EventLogo{socg-logo}
\SeriesVolume{129}
\ArticleNo{0} % <- Your article number goes here
\nolinenumbers %uncomment to disable line numbering
\hideLIPIcs  %uncomment to remove references to LIPIcs series (logo, DOI, ...), e.g. when preparing a pre-final version to be uploaded to arXiv or another public repository
\begin{document}
	
	\maketitle
	
	\begin{abstract}
		We give algorithms to compute the Fr\'{e}chet distance of trees and graphs with bounded tree width. Our algorithms run in $\calO(n^2)$ time for trees of bounded degree, and $\calO(n^2\sqrt{n \log n})$ time for trees of arbitrary degree. For graphs of bounded tree width we show one can compute the \fr{} in FPT (fixed parameter tractable) time.
	\end{abstract}

	\section{Introduction}\label{ch:intro}
	
	The \fr{}, a distance measure for curves introduced by Fr\'{e}chet in \cite{Fr1906}, is a popular measure for comparing polygonal curves. 
	It is defined via homeomorphisms between the parameter spaces of the curves, or intuitively by a man walking his dog and the shortest length of a leash connecting them.

	The \fr{} is well studied. Alt and Godau~\cite{AG95} gave a polynomial time algorithm to compute the \fr{} between two curves, sparking research in many applications, such as character recognition~\cite{charrecon} or navigation on road maps~\cite{Sharma2019}. One can also define the \fr{} between other objects like surfaces \cite{Alt2009} or polygons \cite{BBS10}. Here we study the \fr{} of straight-line embedded graphs, i.e., every vertex of the graph is assigned to a point in the metric space and every edge between two vertices is a straight line segment between the respective points. 
	
	First we observe that two graphs are homeomorphic in the topological sense, if they are isomorphic after their degree~2 vertices have been contracted \cite{BBS10}. That is, a graph homeomorphism induces a graph isomorphism on the contracted graphs, and vice versa.

	Hence in the following we always assume that degree~2 vertices have been contracted and define the \fr{} between %two 
	embedded graphs $G_1=(V_1,E_1)$ and $G_2=(V_1,E_1)$ as:
	\[\delta_F(G_1,G_2) =  \min_{\pi \colon G_1 \mapsto G_2} \ \max_{v \in V_1} \ \|v-\pi(v)\|\] where $\pi \colon G_1 \rightarrow G_2 $ is an isomorphism and $\|v-\pi(v)\|$ is the distance between the points corresponding to $v$ and $\pi(v)$.
	If the graphs do have degree~2 vertices we need to also take into account the \fr{} between contracted edges (embedded as polygonal paths).

	Thus the \fr{} of two graphs is closely related to the (contracted) graphs being isomorph. That is, for isomorphic graphs the \fr{} gives the smallest distance between their embeddings, whereas for non-isomorphic graphs the \fr{} is undefined. Hence the insights on these two notions translate to each other.

	\textbf{Related work:} To the best of our knowledge, the \fr{} of graphs (in this setting) has only been considered by Buchin et al.~\cite{BBS10}.

	They studied the hardness of the \fr{} between surfaces and also discuss the \fr{} of graphs. In particular, they sketched how to compute the \fr{} between two embedded trees. Their idea is to decide whether there is an isomorphism between the trees respecting a given distance by checking this bottom-up in the tree. An isomorphism respects a distance $\delta$ if the distance between every vertex and its image is at most $\delta$. This decider can then be used by a binary search over all possible distances to compute the \fr{}. We expand upon this idea and show that one can compute the \fr{} between two trees even faster when carefully computing the \fr{} between every subtree in a bottom-up fashion.
	
	As one can compute the \fr{} between two trees it seems natural to look at \emph{tree-like} graphs. We say a graph is tree-like if it has bounded tree width, see Section~\ref{ch:graphs}. The definition of the \fr{} given above requires the graphs to be isomorphic. Only recently Lokshtanov et al.  \cite{lpps17} gave a canonization algorithm showing that graph isomorphism for graphs of bounded tree width can be decided in \FPT{} (fixed parameter tractable) time. Their algorithm however does not compute isomorphisms between the graphs, which we are interested in. Grohe et al.  \cite{gns18} gave an algorithm computing all isomorphisms between two input graphs. We alter this algorithm to compute the \fr{} of two embedded graphs with bounded tree width.
	
	\textbf{Results:} In Section~\ref{ch:tree} we give an algorithm computing the \fr{} between two embedded trees. We show that this can be done in $\calO(n^2 \sqrt{d \log d})$ time, where $n$ is the maximum number of vertices and $d$ is the maximum degree of the trees. In Section~\ref{ch:graphs} we show how one can compute the \fr{} of two graphs with $n$ vertices and tree width at most $k$ in $2^{\calO(k \log^c k)}n^{\mathcal{O}(1)} \cdot \log n$ time.

	\section{Contracting the graphs}\label{ch:contract}
	
	Our algorithms assume graphs without vertices of degree~2. Hence we preprocess the two input graphs as follows. We start by finding all paths of maximum length containing only degree~2 vertices in each graph with a DFS. Next we view these paths and the edges as parameterized curves and compute the \fr{} between every such curve of the first and second graph and store them in a 2-dimensional array. The last step is to connect the endpoints of each curve with a marked edge and to delete all inner vertices of the curves. If the two input graphs are trees one can compute all \fr{}s between each of those curves in time $\calO(n^2 \log n)$ and store them in an array size of $\calO(n^2)$. If at most one of the trees has \emph{long} degree~2 paths, i.e., of non-constant length, the computation takes only $\calO(n^2)$ time, as we show next.

	\begin{lemma}
	Let $G=(V,E)$ and $G'=(V',E')$ be two trees with $|V|=|V'|=n$. One can contract all paths of degree~2 vertices and store the \fr{}s between each pair of edges of the contracted graphs in an array of size $\calO(n^ 2)$. This procedure takes $\calO(n^ 2 \log n)$ time, or $\calO(n^ 2)$ time if at least one graph has only degree~2 paths of constant length.
	\end{lemma}
	\begin{proof}
	Given two trees $G=(V,E)$ and $G'=(V',E')$ with $|V|=|V'|=n$. Notice that the paths of degree~2 vertices are pairwise disjoint, except possibly for their endpoints. %Otherwise there would be vertex in such a path with degree higher than 2. 
Thus these paths define a partition of $V$ and $V'$ respectively. 
	Also we can bound the number of such paths in each graph by $\frac{n}{2}$. 
	Let $\ell_1, \dots, \ell_{\frac{n}{2}}$ and $\ell_1', \dots, \ell_{\frac{n}{2}}'$ be the lengths of such paths within $G$ and $G'$ respectively. 
	If there are less than $\frac{n}{2}$ paths the corresponding lengths are $0$. It holds that $\sum_{i=1}^{n/2} \ell_i \leq n$ and $\sum_{i=1}^{n/2} \ell_i' \leq n$. 

          We compute the distance between each pair of edges in the contracted graphs. The easy case is when both edges are non-contracted and we can compute the distance in $\calO(1)$ time. 

If one edge is contracted and the other is not (or is contracted but has constant length), this takes time linear in the length of the contracted edge, which sums up to $\sum_{i=1}^{n/2}(c \cdot \ell'_i) \leq cn$ for a single edge, and hence quadratic time overall.

To compute the \fr{}s between each pair of contracted edges (of non-constant length), i.e. degree~$2$ paths, in $G$ and $G'$ we use the algorithm of Alt and Godau \cite{AG95} to compute the \fr{} of two paths in time $T(\ell_i,\ell_j') = \calO(\ell_i\ell_j'\log (\ell_i\ell_j'))$. 
Hence in total we need time $\sum_{i=1}^{n/2} \sum_{j=1}^{n/2} T(\ell_i,\ell_j') \leq c n^2 \log n$.
Note that this computation time is only necessary if both graphs have long degree~2 paths.

We store the distances in an array of size $O(|E| \cdot |E'|)$. Note that the contracted trees still remain trees, hence this array has size $\calO(n^2)$.  
	All steps take time $\calO(n^2\log n)$  (or $\calO(n^2)$ if there are no long degree~2 paths). 
	\end{proof}
	
	\section{Algorithm for trees}\label{ch:tree}
	
	Computing the \fr{} of two graphs requires the graphs to be isomorphic. It is well known that the isomorphism problem for two trees can be solved in polynomial time \cite{Aho1974}, and therefore we look at trees first. If not stated otherwise the trees have a designated root. If they do not, we use the fact that a tree isomorphism needs to match the centers of the trees.

	We give a brief overview of the algorithm, based on the ideas in~\cite{BBS10}.
	The algorithm computes the \fr{} in a bottom up way, comparing two nodes of same height in every step. Let $T$ and $T'$ be the two input trees, and $t \in T$ and $t' \in T'$ two vertices of the same height and equal degree. % (otherwise the cannot be matched by an isomorphism).
	Let $t_1, \dots , t_i$ and $t'_1, \dots , t_i'$ be the children of $t$ and $t'$. Assume we already have computed the \fr{} between the subtrees rooted at the children.
	To compute the \fr{} between the trees rooted at $t$ and $t'$ we do the following. We create a new bipartite graph $B$. For every subtree of $t$ and $t'$ there is a vertex in $B$. We combine every vertex corresponding to a subtree of $t$ with every vertex corresponding to a subtree of $t'$ with a weighted edge. The weight of the edge is the maximum of the \fr{} between the subtrees, and the \fr{} of the edges connecting the subtrees to $t$ or $t'$.
	
	Now we find a \emph{bottleneck matching} for this graph. A bottleneck matching is a perfect matching such that the maximum weight of its edges is minimal for all perfect matchings possible.
	The \fr{} of the trees rooted at $t$ and $t'$ is either the maximum weight of the found bottleneck matching or the distance of $t$ and $t'$. We store the correct value in a two dimensional array. If we could not find a bottleneck matching the trees are not isomorphic and thus their \fr{} is undefined. In this case we store $\infty$. If the subtrees are empty we store $0$.
	
	The algorithm iterates over both trees doing the above computation for each pair of nodes of same height. Afterwards the algorithm returns the value stored for the roots of the trees.
	
	Figure \ref{fig:bottleneck} illustrates the computation of a bottleneck matching, given the two trees in (a) as input. Consider the computation of the \fr{} of the trees rooted at $u$ and $u'$. The graph $B$ is shown in (b) and the distances are given in the table in (c). In (d) we see the computed bottleneck matching with maximum \fr{} of $1$.
	
	\begin{figure}[ht]
		\begin{subfigure}[b]{0.9\linewidth}
			\centering
			\includegraphics[width=0.9\textwidth]{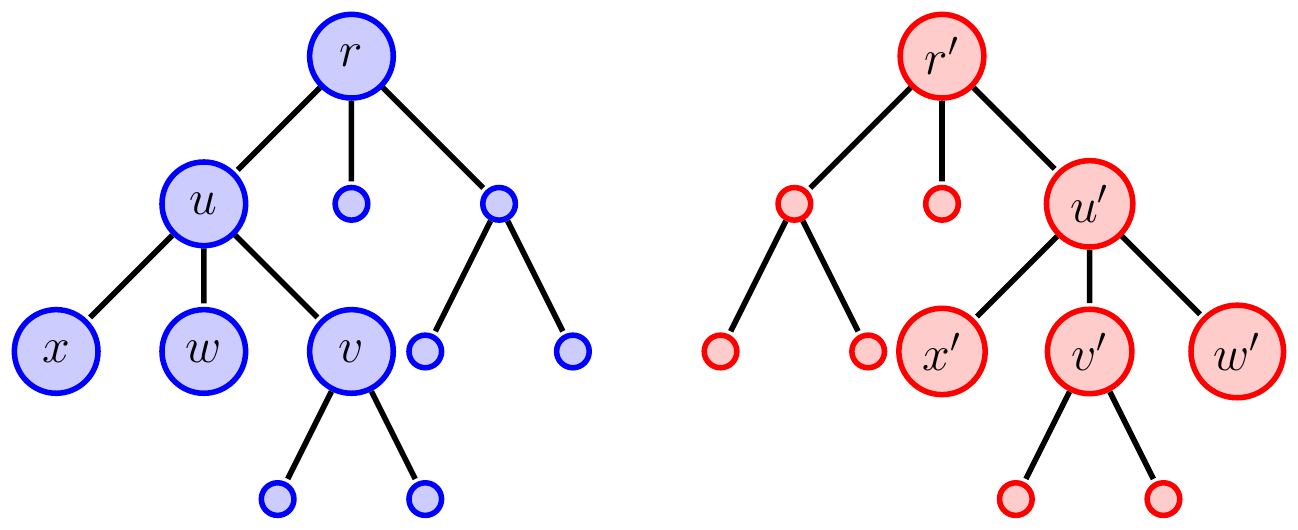}
			\subcaption{}
		\end{subfigure}\hfill
		
		\vspace{5ex}
		
		\begin{subfigure}[b]{0.33\linewidth}
			\centering
			\includegraphics[width=0.75\textwidth]{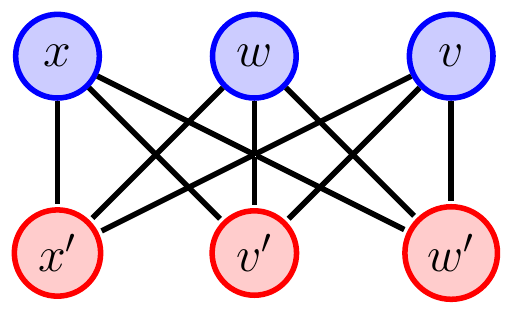}
			\subcaption{ }
		\end{subfigure}
		\begin{subfigure}[b]{0.3\linewidth}
			\centering
			\begin{tabular}{|c|c|c|c|}
				\hline
				& $v$ & $w$ & $x$ \\
				\hline
				$v'$ & $1$ & $2$  & $2$ \\
				\hline
				$w'$ & $2$ & $1$ & $2$ \\
				\hline
				$x'$ & $2$ & $2$ & $1$  \\
				\hline
			\end{tabular}
			\subcaption{ }
		\end{subfigure}
		\begin{subfigure}[b]{0.33\linewidth}
			\centering
			\includegraphics[width=0.75\textwidth]{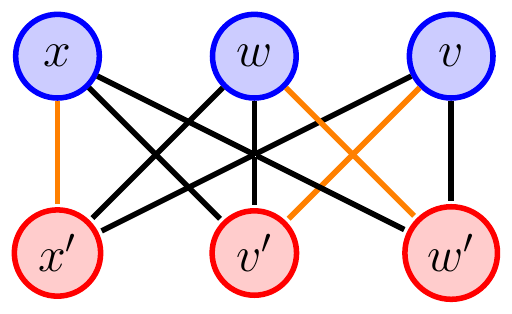}
			\subcaption{ }
		\end{subfigure}
		\caption{Two trees~(a), the bipartite graph $B$~(b), the weights~(c), and bottleneck matching~(d)}
		\label{fig:bottleneck}
	\end{figure}
	
	It remains to compute such a bottleneck matching. To do this we store all edge weights in a sorted array and perform a binary search. The perfect matching for the distance found by this search is our desired bottleneck matching.
	
	\bigskip
	
	The running time of this algorithm is as follows. We use the algorithm of Alt et al. \cite{alt91} to compute the perfect matchings. This means we can bound the running time of one bottleneck matching computation by $\calO\left(d^2\sqrt{d \log d}\right)$ time, with $d$ being the degree of the trees. %WE analyze the running time.
	
	The running time of the algorithm boils down to the time we compute a bottleneck matching. To compute such a matching we search the smallest edge weight under which there is a perfect matching. Since the graphs we do this for are of size $\calO(d)$ we need time $\calO \left(d^2 \sqrt{\frac{d}{\log d}}\right)$ to compute a perfect matching using the algorithm of Alt et al. \cite{alt91}. We perform a binary search over all weights using this subroutine, thus we need time $\calO(d^2 \sqrt{d \log d})$ for that. Next we analyze how often we compute a bottleneck matching. Let $h$ be the height of the trees. Then we compute $\sum_{i=1}^h {\left(d^i\right)^2}$ bottleneck matchings. Thus we need time $c \cdot d^2 \sqrt{d\log d} \cdot \sum_{i=1}^h{\left(d^i\right)^2} \leq c' n^2 \sqrt{d \log d} \in \calO (n^2 \sqrt{d \log d})$, for positive constants $c,c'$. This yields:
	
	\begin{theorem}\label{thm:trees}
		The \fr{} of two embedded trees can be computed in time $\calO\left(n^2\sqrt{n \log n}\right)$, if at most one of the trees has long degree~2 paths. If the trees have bounded degree the computation only takes $\calO(n^2)$ time. If both trees have long degree~2 paths, an extra $\calO\left(n^2 \log n\right)$ time is needed for preprocessing.
	\end{theorem}
	
	Note that our computation requires the metric of the space to be computed in constant time. If that is not the case the running time for this algorithm increases due to the distance computation.
	Next we generalize this result to graphs that are not trees but are tree-like in their structure.
	
	\section{Algorithm for graphs with bounded tree width}\label{ch:graphs}
	
	Here we consider a larger class of graphs. It is known that many hard problems for graphs can be solved easily on trees, such as 3-coloring problem and finding a vertex cover. Both these problems can be solved in linear time on a tree. One can ask if these problems are easy for tree-like graphs. A measurement for such a likeliness is the \emph{treewidth}. To define the treewidth we first introduce \emph{tree decompositions}. Intuitively a tree decomposition of a graph $G$ represents the vertices and edges of $G$ as subgraphs inside a tree. Formally:
	\begin{definition}
		Let $G=(V,E)$ be a graph. Let $T$ be a tree and $\beta$ a function mapping all $t \in T$, the nodes, to subsets of $V$. We call $(\beta,T)$ a tree decomposition of $G$ if:
		\begin{enumerate}
			\item for each vertex $v \in V$ it holds that all subsets $\beta(t) \in T$, called bags, containing $v$ induce a nonempty connected subtree of $T$, and
			\item for each edge $(u,v) \in E$ there is one $t \in T$, such that the bag $\beta(t)$ contains $u$ and $v$.
		\end{enumerate}
	\end{definition}
	
	The width of such a tree decomposition is the size of its largest bag $-1$. The treewidth of $G$ is the minimal width among all its tree decompositions. We subtract $1$ from the bag size so trees have treewidth $1$.
	
	Figure~\ref{fig:tdc} shows a graph $G$ (a) and two valid tree decompositions of $G$ (b) and (c). Decomposition (c) is a trivial tree decomposition consisting only of one node and (b) is a tree decomposition, that realizes the treewidth of $G$
	
	\begin{figure}[ht]
		\begin{subfigure}[b]{0.33\linewidth}
			\centering
			\includegraphics[width=0.75\textwidth]{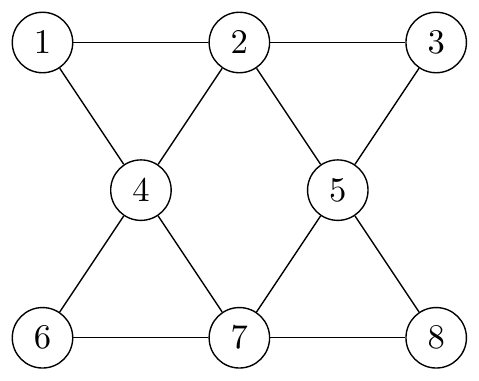}
			\subcaption{ }
		\end{subfigure}\hfill
		\begin{subfigure}[b]{0.33\linewidth}
			\centering
			\includegraphics[width=0.75\textwidth]{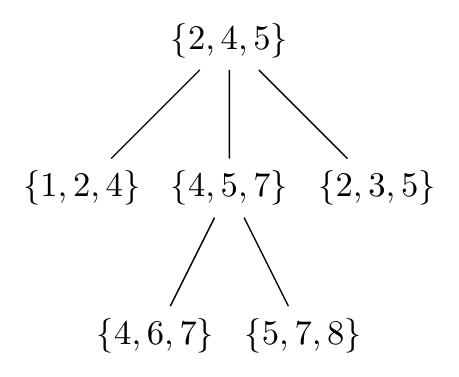}
			\subcaption{ }
		\end{subfigure}
		\begin{subfigure}[b]{0.33\linewidth}
			\centering
			\includegraphics[width=0.75\textwidth]{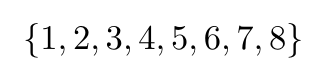}
			\subcaption{ }
		\end{subfigure}
		\caption{The graph $G$ (a) and two tree decompositions (b) and (c),  of $G$}
		\label{fig:tdc}
	\end{figure}
	
	Tree decompositions play an important role in the context of parameterized algorithms. An algorithm is parameterized, if its running time does not only depend on the input size $n$ but also on a parameter $k$. A problem is fixed parameter tractable, if there is a algorithm with running time $f(k) \cdot \text{poly}(n)$ for it. Many graph problems have algorithms with tree width as parameter, see \cite{downey2012, flum2006}. Typically these algorithms use dynamic programming over a tree decomposition of the graph. Our algorithm is based on the algorithm of Grohe et al.~\cite{gns18}. We extended their algorithm, such that we can compute all isomorphisms between two embedded graphs that respect a distance limit.
	Using this algorithm we conduct a binary search on every possible distance between nodes of $G_1$ and $G_2$. The \fr{} of the two graphs is the smallest distance found that is respected by an isomorphism.
	
	We describe our algorithm.
	Let $G_1=(V_1,E_1)$ and $G_2=(V_2,E_2)$ be two embedded graphs, and $\delta\geq 0$ a distance limit. We start by preprocessing the graphs in a similar fashion as in the case of trees, see Section \ref{ch:tree}. Since the number of edges of each graph might exceed $n$ we need more time to do this, but clearly one can compute all the distances in polynomial time. Next we compute an initial tree decomposition using the techniques of Leimer \cite{lei93}.
	The resulting decomposition is known to be isomorphism invariant. We carefully refine these tree decompositions in a bottom up way using dynamic programming.
	In every step we maintain the invariant that the tree decompositions are isomorphism invariant, the size of the resulting bags is not too high, and we get crucial information about the structure of both graphs.

	We now use the resulting tree decomposition to compute the isomorphisms respecting $\delta$ in a bottom up way. Suppose we are looking at two nodes $t_1 \in T_1$ and $t_2 \in T_2$, and have already computed all isomorphisms respecting $\delta$ between the subtrees of $t_1$ and $t_2$ rooted at their respective children.
	To compute the isomorphisms between the graphs induced by the trees rooted at $t_1$ and $t_2$ Grohe et al. \cite{gns18} use two abstractions. 
	
	Let $A$ be a set, $\varphi$ be a bijective mapping $A \rightarrow |A|$ and $\Theta$ be a subgroup of Sym($|A|$), the symmetric group of $|A|$, then $\varphi \Theta$ is a \emph{labeling} coset. Let $\psi \Lambda \subseteq \varphi \Theta$. If $\psi \Lambda$ forms a labeling coset we call it a \emph{labeling subcoset} of $\varphi \Theta$ ($\psi \Lambda \precsim \varphi \Theta$). The first abstraction is called \emph{coset-hypergraph-isomorphism} and is defined as follows \cite{gns18}:
	
	An instance of \emph{coset-hypergraph-isomorphism} is a 8-tuple $\mathcal{I}=(V_1,\mathcal{S}_1,V_2,\mathcal{S}_2, \chi_1, \chi_2, \mathcal{F}, f)$, such that for $i \in \{1,2\}$:
	
	\begin{enumerate}
		\item $\mathcal{S}_i \subseteq 2^{V_i}$,
		\item $\chi_i \colon \mathcal{S}_i \rightarrow \mathbb{N}$ is a coloring,
		\item $\mathcal{F} = \{\Theta_S \mid \Theta_S \precsim $ Sym$(S), S \in \mathcal{S}_1\}$ and
		\item $f = \{\tau_{S_1S_2} \colon S_1 \rightarrow S_2 \mid S_1 \in \mathcal{S}_1, S_2 \in \mathcal{S}_2, \chi_1(S_1) = \chi_2(S_2) \}$, such that
		\begin{itemize}
			\item[\textbf{a}] every $\tau_{S_1S_2}$ is bijective, and
			\item[\textbf{b}] for every color $j$, every $S_1, {S'}_1 \in \chi_1^{-1}(j)$, every $S_2, {S'}_2 \in \chi_2^{-1}(j)$ and every $\theta \in \Theta_{S_1}, \theta' \in \Theta_{{S'}_1}$ it holds that
			
			$\theta \tau_{S_1S_2}(\tau_{{S_1}'S_2})^{-1}\theta'\tau_{{S_1}'{S_2}'}$ is an element of the labeling coset $\Theta_{S_1}\tau_{{S_1}'S_2}$.
		\end{itemize}
	\end{enumerate}

The second abstraction used is called \emph{multiple-colored-coset-isomorphism} and is defined as follows \cite{gns18}:

An instance of \emph{multiple-colored-coset-isomorphism} is a 6-tuple $\mathcal{J} = (V_1,V_2,\chi_1,\chi_2,\mathcal{F},f)$ such that for $i \in \{1,2\}$

\begin{enumerate}
	\item $\chi_i \colon [t] \rightarrow \mathbb{N}$ is a coloring,
	\item $\mathcal{F} = \{\Theta_{i'} \mid \Theta_{i'} \precsim $ Sym$(V), i' \in [t]\}$ and
	\item $f = \{\tau_{i',j'} \colon V_1 \rightarrow V_2 \mid i',j' \in [t], $ with $ \chi_1(i') = \chi_2(j') \}$, such that
	\begin{itemize}
		\item[\textbf{a}] every $\tau_{i',j}$ is bijective, and
		\item[\textbf{b}] for every color $'$, every $j_1,{j_1}' \in \chi_1^{-1}(j)$, every $j_2,{j_2}' \in \chi_1^{-1}(j)$ and every $\theta \in \Theta_{j_1}, \theta' \in \Theta_{{j_1}'}$ it holds that
		
		$\theta \tau_{j_1,j_2}(\tau_{{j_1}',j_2})^{-1}\theta'\tau_{{j_1}'{j_2}'}$ is an element of the labeling coset $\Theta_{j_1}\tau_{{j_1}'j_2}$.
	\end{itemize}
\end{enumerate}
	
	Grohe et al. \cite{gns18} use these abstractions to compute all isomorphisms between the graphs using the tree decompositions in a bottom up way using dynamic programming. The abstractions encode the following: $V_i$ are the nodes inside the bags, $\mathcal{S}_i$ are the intersection of the bags with their respective children. For all $S_i \in \mathcal{S}_1$ and $S_j \in \mathcal{S}_2$ the colorings $\chi_1$ and $\chi_2$ are defined as follows: $\chi_1(S_i) = \chi_2(S_j) \Leftrightarrow S_i$ and $S_j$ are isomorphic. The sets $\mathcal{F}$ and $f$ contain all isomorphisms that are computed for the subtrees. The algorithm then computes all isomorphisms between the graphs induced by the bags and their subtrees. 
	
	We alter this in the following way: two sets $S_i, S_j$ are assigned the same color, iff there is an isomorphism between them respecting the given distance limit. Note that properties \textbf{4.b} and \textbf{3.b} require us to have all isomorphisms between two such sets stored. These isomorphisms might not necessarily respect the distance. In this case we simply mark the isomorphism as not sufficient to our needs. The result of this is a representation of isomorphisms between the two graphs. If this representation only contains marked isomorphisms we treat the representation as empty, meaning the graphs are not isomorphic under the distance. In all other cases the representation contains all isomorphisms between the graphs that respect the distance limit.
		
	This means one can decide whether the two graphs have \fr{} at most $k$ in time $2^{\calO(k \log^c k)}n^{\mathcal{O}(1)}$, where $n$ is the number of vertices of one graph, and $c$ is a positive constant.

	With this we perform the binary search for the minimum distance between two vertices that is respected by at least one isomorphism and get the following result:
	\begin{theorem}\label{thm:graphs}
		The \fr{} between two graphs with $n$ nodes and treewidth at most $k$ can be computed in time $2^{\calO(k \log^c k)} \cdot n^{\mathcal{O}(1)} \cdot \log n$.
	\end{theorem}
	
	\section{Conclusion} % and Discussion}
	
	We have shown how to compute the \fr{} of two rooted trees and of two graphs with bounded treewidth.
	Several interesting question remain. Since a polynomial time algorithm for computing the \fr{} of two graphs could be used to decide whether the graphs are isomorphic, we do not expect a polynomial time algorithm for general graphs. But perhaps there are other parameters in which the \fr{} of two graphs is fixed parameter tractable.\bigskip

	\hspace{-5.5mm}\textbf{Acknowledgements.} This work is based on the BSc thesis and study research project by the third author Alexander Neuhaus. 
	
	\nolinenumbers
	\footnotesize
	\bibliography{references}

\end{document}